%% file: iswcs08.tex
\documentclass{IEEEtran}
\usepackage{times,amsmath,epsfig}
\usepackage{psfrag}
\usepackage{graphicx}
\usepackage{flushend}
\usepackage{microtype}
\usepackage{type1cm}
\usepackage{url}
\usepackage{amsthm}
\usepackage{amssymb}
\usepackage{bm}
\usepackage{xspace}
\usepackage{xcolor}
\usepackage{float}
\usepackage[caption=false]{caption}
\usepackage[font=footnotesize]{subfig}

\newtheorem{theorem}{Theorem}
\newcommand{\safemath}[2]{\newcommand{#1}{\ensuremath{#2}\xspace}}

\safemath{\cndis}{\mathcal{CN}}

\safemath{\mysetN}{\mathbb{N}}
\safemath{\mysetC}{\mathbb{C}}
\safemath{\mysetR}{\mathbb{R}}
\safemath{\mysetRP}{\mysetR^\oplus}
\safemath{\mysetZ}{\mathbb{Z}}

\safemath{\ltwoR}{\mathcal{L}^2(\mysetR)}

\safemath{\eff}{{\mathrm{eff}}}
\safemath{\heff}{h_T}
\safemath{\Heff}{H_T}
\safemath{\htrunc}{\tilde{h}}
\safemath{\dspread}{D_s}
\safemath{\toffset}{\epsilon}
\safemath{\powerPenalty}{\mathcal{P}}
\safemath{\nanos}{\mathrm{ns}}
\safemath{\mhz}{\mathrm{MHz}}

\newcommand{\de}{\,\mathrm{d}}
\newcommand{\lefto}{\mathopen{}\left}
\newcommand{\parmark}[1]{\hspace{-\parindent}\textit{#1.}}

\DeclareMathOperator{\fouop}{F}

\title{Theoretical Analysis of the Energy Capture in Strictly Bandlimited
Ultra-Wideband Channels}
\author{%
{Georg B\"ocherer{\small $~^{1}$}, Daniel Bielefeld{\small $~^{2}$},
Rudolf Mathar{\small $~^{3}$} }%
\vspace{1.6mm}\\
\fontsize{10}{10}\selectfont\itshape
Institute for Theoretical Information Technology, RWTH Aachen
University\\
D-52056 Aachen, Germany\\
\fontsize{9}{9}\selectfont\ttfamily\upshape
$^{1}$\,boecherer@ti.rwth-aachen.de\\
$^{2}$\,bielefeld@ti.rwth-aachen.de\\
$^{3}$\,mathar@ti.rwth-aachen.de%
}
\begin{document}
\maketitle

\input{abstract}

\input{introduction}

\input{model}

\input{analysis}

\input{simulation}

\input{conclusions}

\input{acknowledgement}

\bibliographystyle{IEEEtran}

\bibliography{IEEEabrv,Literatur}

\end{document}

%% file: abstract.tex
\begin{abstract}
The frequency selectivity of wireless communication channels can
be characterized by the delay spread $D_s$ of the channel impulse
response. If the delay spread is small compared to the bandwidth
$W$ of the input signal, that is, $D_s W\approx 1$, the channel
appears to be flat fading. For $D_s W\gg 1$, the channel appears
to be frequency selective, which is usually the case for wideband
signals. In the first case, small scale synchronization with a
precision much higher than the sampling time $T=1/W$ is crucial to guarantee
the maximum capture of energy at the receiver. In this paper, it is shown by
analytical means that this is different in the wideband regime. Here
synchronization with a precision of $T$ is sufficient and small scale
synchronization cannot further increase the captured energy at the receiver.
Simulation results show that this effect already occurs for $W>50$MHz for the
IEEE 802.15.4a
channel model.
\end{abstract}

%% file: introduction.tex
\section{Introduction}
\label{sec:intro}
Ultra-wideband communication offers the attractive possibility to achieve high
data rates with low transmission power and is a candidate for the
air interface of next generation wireless personal area
networks~\cite{Arslan2006}. 

The transmitted ultra-wideband signal is
usually subject to frequency selective fading due to multi-path
propagation. In literature, the involved channel has been modeled in different
ways.
Physically inspired approaches aim to give a tapped delay-line model of the
propagation environment in terms of a time-continuous impulse response which
consists of the superposition of arbitrarily delayed and scaled versions of the
signal emitted by the sender. In \cite{Win1998}, the authors refer to the
number of delayed and scaled versions assumed at the receiver as the
\textit{diversity
level} of the receiver and they observe that with an increasing diversity level,
the amount of captured energy increases.

The authors in \cite{Schuster2007} use a more
general model. They assume a strictly bandlimited transmitted
signal and model the channel by the statistics of the taps as observed by the
receiver after appropriate lowpass filtering and sampling of the received
signal with a sampling time according to the sampling theorem. The
obtained time-discrete channel model consists of a finite number
of channel taps. A similar model is used in~\cite{Telatar2000}, where
the sampling time is referred to as the \textit{system resolution}. In
these models, the observed channel taps depend heavily on the
chosen time instants of the sampling process. For different timing
offsets between transmitter and receiver, corresponding to
different sampling time instants, the receiver might observe
different realizations of the channel, although the physical
channel remains the same. In the following, the different
realizations of channel taps for the same physical channel in the
time-discrete model are referred to as \emph{channel candidates}. For a
low number of channel taps, normally corresponding to a low signal
bandwidth, the unknown timing offset can lead to a significant
degradation of the amount of energy that can be captured by the receiver,
since different observed channel candidates have in general different channel
gains. To account for this in practical communication systems, synchronization
algorithms are employed to conduct a precise estimate of the
timing offset and to synchronize the time references of
transmitter and receiver~\cite{Meyr1998}.

In this paper, we analytically investigate the problem of different channel
candidates in the wideband limit. We show that in
this case, all channel candidates have the same channel gain if a
large scale synchronization in the order of a sampling time
interval is assumed. This is not obvious: although the system resolution
increases with increasing bandwith, also the degrees of freedom of the effective
channel increases linearly with the bandwidth \cite{Schuster2007}, which could
cancel out the effect of increasing resolution. This is illustrated in
Figure~\ref{fig:visualization}. Our result implies that for a
large bandwidth of the input signal, it is not necessary to consider the delay
between transmitter and receiver by an additional term in the discrete
time channel model of~\cite{Schuster2007}, since it would not make the model
more representative.

To validate our results and to obtain numerical
values for the amount of the performance degradation for a realistic scenario,
simulations were conducted using the IEEE 802.15.4a channel model. The results
indicate that for a narrow communication bandwidth ($\leq16$MHz), a missing
small scale synchronization can lead to an energy capture degradation of
more than $25\%$. For increasing bandwidth, the degradation vanishes. If we
interpret the corresponding number of considered channel taps as the diversity
level of the receiver, this result is in accordance with the
observation in \cite{Win1998}.

The remainder of the paper is organized as follows. In
Section~\ref{sec:model}, a mathematical description of the
considered problem is given. Based on this, we investigate the problem
analytically in Section~\ref{sec:analysis}. In Section~\ref{sec:simulation}, we
describe the simulation procedure and discuss the obtained results.

\begin{figure*}
\small
\centering
\subfloat[Sampling in narrowband]{
\psfrag{POWER}{Power}
\psfrag{TIME}{Time}
\includegraphics[width=0.48\linewidth]{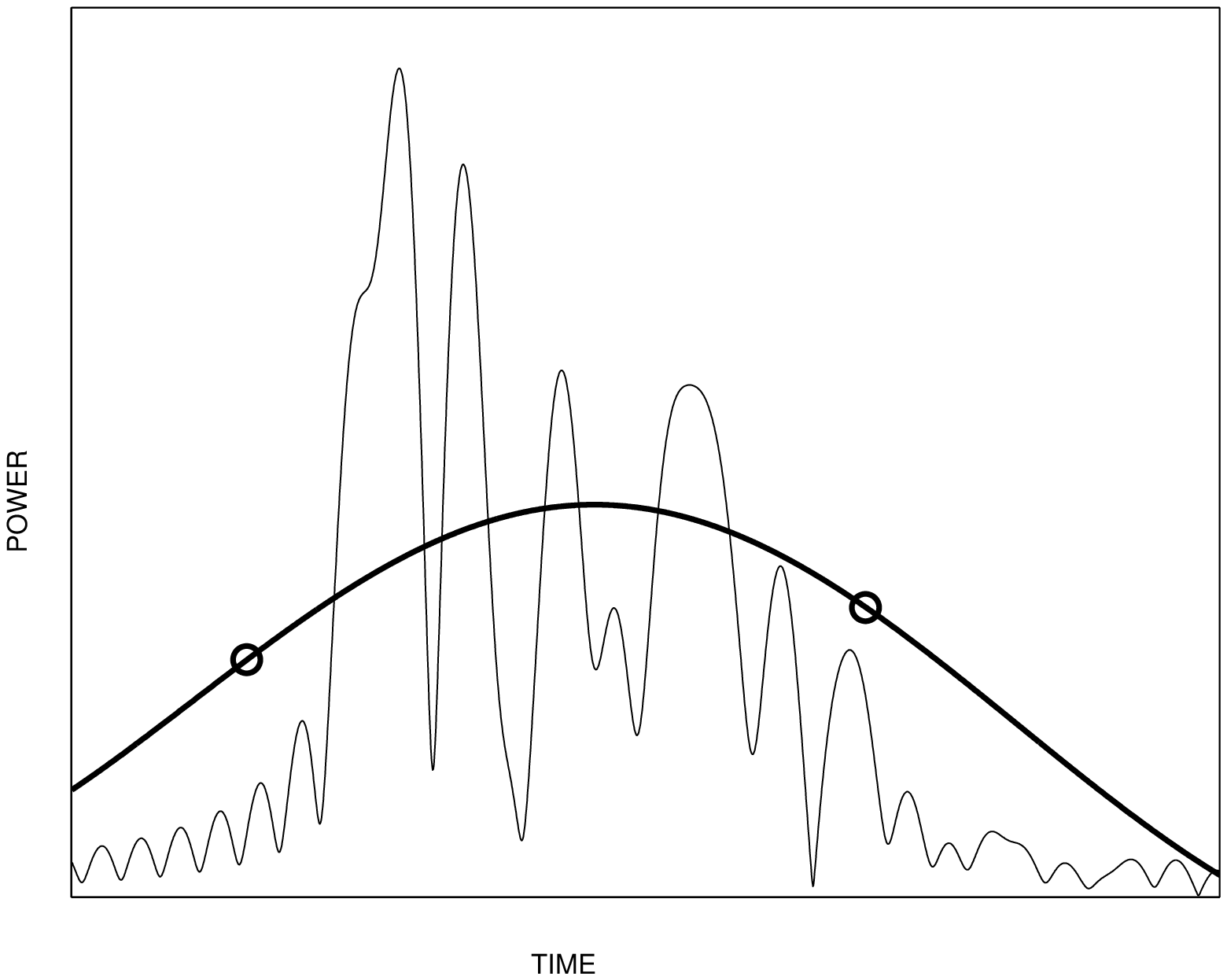}
\label{fig:visualization1}
}
\subfloat[Sampling in wideband]{
\psfrag{POWER}{Power}
\psfrag{TIME}{Time}
\includegraphics[width=0.48\linewidth]{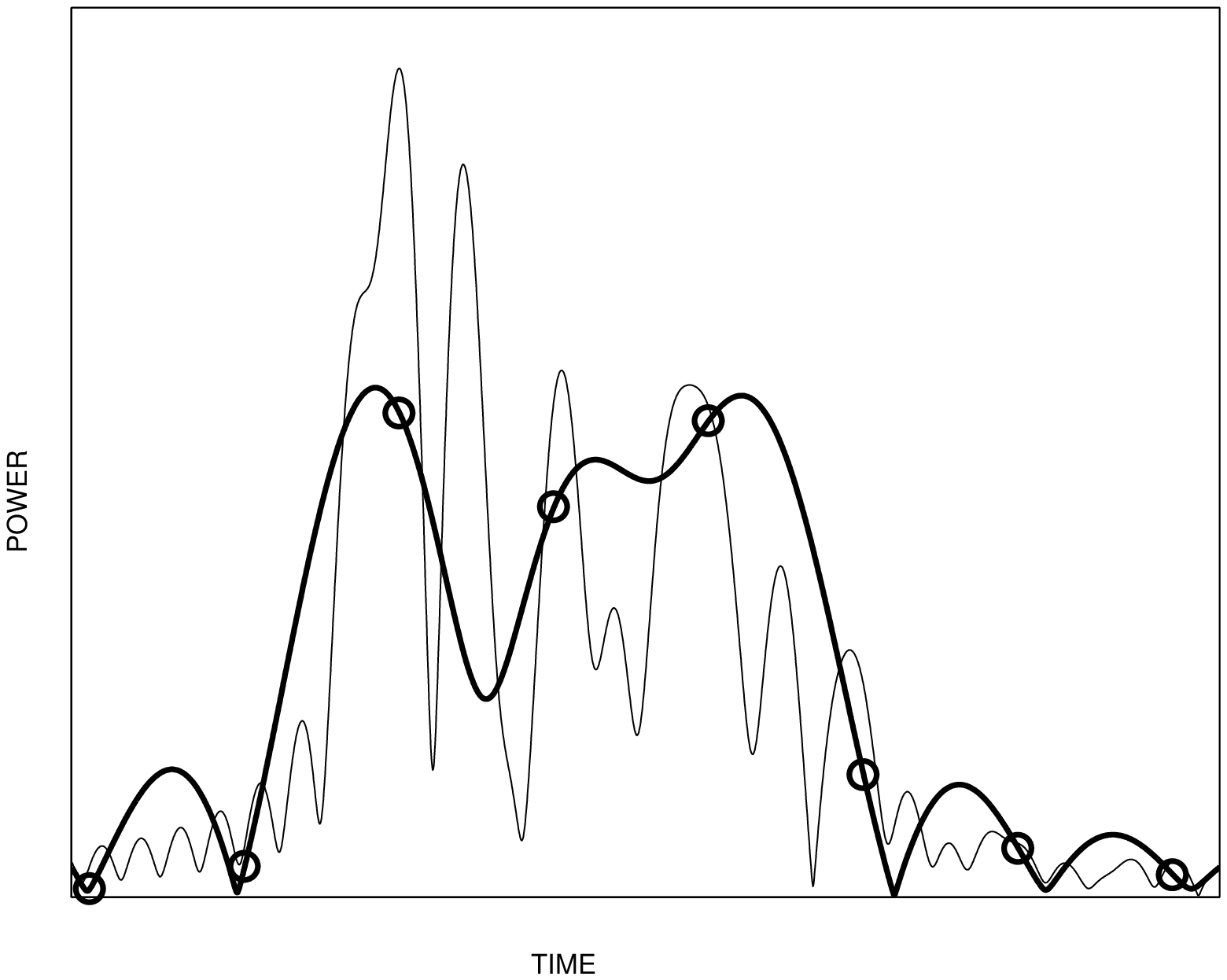}
\label{fig:visualization2}
}
\caption{An impulse response of bandwidth $128$MHz (thin line) is in
subfigure (a) filtered by a lowpassfilter of bandwidth $8$MHz (bold line), and
then sampled (circles). In subfigure (b), the same impulse response is filtered
by a lowpassfilter of bandwidth $64$MHz and then sampled. As can be observed,
both system resolution and degrees of freedom of the filtered impulse response
increase with increasing filter bandwidth.}
\label{fig:visualization}
\end{figure*}

%% file: model.tex
\section{System Model}\label{sec:model}
We consider a linear time-invariant (LTI) system that is specified
by a complex continuous time impulse response $h$, which we assume in the
following to be in $\ltwoR$, the Hilbert space of energy limited signals. The
impulse
response is of finite length, that is, there exists some delay spread
$\dspread>0$ such
that 
\begin{align}
\forall \tau\notin[0,\dspread]\colon h(\tau)=0.
\end{align}
The input-output (I/O) relation between a transmitted signal $s$ and the
corresponding received signal $y$ is then given by the convolution
\begin{align}
 y(t)=\int\limits_{-\infty}^\infty
h(\tau)s(t-\tau)\de\tau.\label{eq:continuousIO}
\end{align}
If the input signal is bandlimited to $W,$ then the output signal is also
bandlimited to $W$ and we can define the \emph{effective impulse response}
$\heff$ of the LTI system as the result of lowpass filtering $h$ by
a unit-gain lowpass filter of bandwidth $W=1/T$. It is important to note that
the effective impulse response has no longer a finite delay spread, so we have
in general $\heff(\tau)\neq 0$ for $\tau\in\mysetR$. Applying the
sampling theorem, the system can be described by the discrete time I/O relation
given by
\begin{align}
 y[k]=\sum\limits_{l=-\infty}^\infty \heff[l]s[k-l].\label{eq:discreteIO}
\end{align}
Although we have in general $h_T[l]\neq 0$ for $l\in\mysetZ$, it is infeasible
in
practice to treat an infinite number of channel taps. We therefore truncate the
I/O relation and consider only a finite number $L$ of channel taps. The
truncated I/O relation is given by
\begin{align}
 y[k]=\sum\limits_{l=0}^{L-1} \heff[l]s[k-l].\label{eq:truncatedIO}
\end{align}
For the number $L$, an appropriate number corresponding to the delay spread
$\dspread$ should be chosen.

The relation \eqref{eq:truncatedIO} is a standard model for frequency selective
channels, see \cite{Tse2005}. The authors of \cite{Schuster2007} modeled an UWB
indoor environment by \eqref{eq:truncatedIO} and based their measurements on
this model.

We will take a closer look at how \eqref{eq:truncatedIO} is related to
\eqref{eq:continuousIO}. The discrete time I/O relation results from two
sampling processes, one at the sender and one at the receiver. We assume that
the drift between the two sampling clocks is already compensated. What remains
is a timing offset between the two sampling clocks. We model this offset
by including an arbitrary timing offset $d$ between sender and receiver in our
model. The timing offset $d$ can be written as 
\begin{align}
d=\left\lfloor \frac{d}{T}\right\rfloor T+\delta 
\end{align}
with $\delta\in[0,T)$. We assume that the system has already performed a large
scale acquisition of the timing and knows $\lfloor d/T\rfloor$. Without loss of
generality, we can therefore assume $\lfloor d/T\rfloor=0$. The samples are then
given by
\begin{align}
 y[k]=y(kT-\delta),\,\heff[l]=\heff(lT-\delta),\,s[k]=s(kT)
\end{align}
where $t$ is the timing reference at the sender and $t-\delta$ is the timing
reference at the receiver. For a certain timing offset $\delta$, the channel
candidate, which will be estimated by the receiver, is given by
\begin{align}
\htrunc_T^{(\delta)}[l]=\left\{\begin{array}{ll}
          \heff(lT-\delta),&0\leq l<L\\
	0,&\text{otherwise}.
         \end{array}\right.\label{eq:channelCandidate}
\end{align}
We set the number of channel taps $L$ equal to
\begin{align}
 L=\left\lfloor\frac{\dspread}{T}\right\rfloor.
\end{align}
This assignment is to a certain extend arbitrary. The trade-off between system
performance and complexity may lead to other values for $L$ in practice.

The small scale timing synchronization at the receiver now consists in finding
$\delta$ such that the corresponding channel candidate
\eqref{eq:channelCandidate} used in \eqref{eq:truncatedIO} represents
\eqref{eq:continuousIO} in the best possible way. To quantify the quality of a
certain channel candidate \eqref{eq:channelCandidate}, we look at the overall
channel gain $\lVert\htrunc_T^{(\delta)}\rVert^2$, which is given by
\begin{align}
 T\lVert\htrunc_T^{(\delta)}\rVert^2=\sum\limits_{l=0}^{L-1}T\lvert
\heff(lT-\delta)\rvert^2.
\end{align}
Depending on the bandwidth $W$ and the corresponding sampling time $T=1/W$, the
channel gain varies for different $\delta\in[0,T)$. As we will see in the
following sections, the variance of the channel gain goes to zero for $W$
tending to infinity.

%% file: analysis.tex
\section{Analysis}\label{sec:analysis}
For the family of channel candidates \eqref{eq:channelCandidate}, we state the
limit property of the overall channel gain as a theorem.
\begin{theorem}\label{theo:channelGain}
Let the channel candidates as given in \eqref{eq:channelCandidate} be of
bandwidth $W=1/T$. Then, for every $\delta\in[0,1/W)$, the overall channel gain
converges to the maximum value given
by $\lVert\htrunc_T^{(\delta)}\rVert^2=\lVert h\rVert^2$ when the
bandwidth $W$ tends to infinity.
\end{theorem}
Before we can give the proof of the theorem, we state Plancherel's
Theorem. It can be found e.g. in \cite{Paley1934}.
\begin{theorem}(Plancherel) \label{plancherel}
For any $x\in\ltwoR$ (the Hilbert space of
energy limited signals), there exists a function $\fouop\{x\}\in\ltwoR$, such
that 
\begin{align}
\lim_{T\rightarrow\infty}\int\limits_{-\infty}^\infty\Bigl\lvert
\fouop\{x\}(f)-\int\limits_{-\frac{T}{2}}^{\frac{T}{2}}x(t)e^{-j2\pi ft}\de
t\Bigr\rvert^2\de f = 0
\end{align}
and
\begin{align}
\lim_{W\rightarrow\infty}\int\limits_{-\infty}^\infty\Bigl\lvert
x(t)-\int\limits_{-\frac{W}{2}}^{\frac{W}{2}}\fouop\{x\}(f)e^{j2\pi ft}\de
f\Bigr\rvert^2\de t = 0.
\end{align}
The function $\fouop\{x\}$ is called the \emph{Fourier transform} of $x$.
\end{theorem}
We assume in the following that for all functions $x$ of interest, the integral
\begin{align}
 \int\limits_{-\infty}^\infty x(t)e^{-j2\pi ft}\de t\label{eq:fourierIntegral}
\end{align}
exists. The Fourier transform $\fouop\{x\}$ is then given by
\eqref{eq:fourierIntegral}. We will need the following convention in the proof
of Theorem~\ref{theo:channelGain}. We define the effective impulse response
$\heff$ point-wise by the inverse Fourier transform of $\fouop\{\heff\}$ given
by
\begin{align}
 \heff(\tau)=\int\limits_{-\infty}^{\infty}\fouop\{\heff\}(f)e^{j2\pi f\tau}\de
f.\label{eq:continuous}
\end{align}
It can easily be shown that with this definition, $\heff$ is continuous.

\begin{proof}[Proof of Theorem \ref{theo:channelGain}]
The energy of $\htrunc_T^{(\delta)}$ is bounded from above by
\begin{align}
 \lVert h\rVert^2&\geq \lVert\heff\rVert^2\\
&\geq T\lVert\htrunc_T^{(\delta)}\rVert^2\\
&=\sum\limits_{l=0}^{\lfloor\frac{D_s}{T}\rfloor-1}
T\bigl|h_T(lT-\delta)\bigr|^2\label{eq:objectOfInterest}
\end{align}
where $\delta\in[0,T)$. We show that \eqref{eq:objectOfInterest} converges to
$\lVert h\rVert^2$ for $T\rightarrow 0$. Calculating first the
limit for $\delta\rightarrow 0$ and then the limit for $T\rightarrow 0$ would
lead to the desired result, but this implies not necessarily that
\eqref{eq:objectOfInterest} converges to the same value if $\delta$ and $T$
tend jointly to zero along any path with $\delta<T$, but it is the latter we
have to show.

We introduce the auxiliary parameter $T'$ and write
\begin{align*}
\sum\limits_{l=0}^{\lfloor\frac{D_s}{T'}\rfloor-1}
T'\bigl|h_T(lT'-\delta)\bigr|^2.
\end{align*}
The introduction of the parameter $T'$ can be interpreted as a separation of the
sampling process from the lowpass filtering process.

We now approximate $\lVert h\rVert^2$ in three steps.

\parmark{1. Approximation}
From Theorem \ref{plancherel}, it follows that the integral
\begin{align}
  \int\limits_{-\infty}^\infty \lvert h(\tau)-\heff(\tau)\rvert^2 \de \tau
\end{align}
goes to $0$ for $T\rightarrow 0$. Since the support of $h$
is $[0,D_s]$, also 
\begin{align}
  \int\limits_0^{D_s} \lvert h(\tau)-\heff(\tau)\rvert^2 \de{\tau}
\end{align}
goes to $0$. For any $v,w\in\ltwoR$, the reverse triangular inequality
states that 
\begin{align}
  \lvert\lVert v\rVert-\lVert w\rVert\rvert\leq \lVert
v-w\rVert.\label{eq:reverseTriangular}
\end{align}
It follows that for any $\epsilon_1>0$ there exists a $T_1$ such that
\begin{align}
  \left\rvert\int\limits_0^{D_s} \lvert h(\tau)\rvert^2
    \de \tau-\int\limits_0^{D_s} \lvert \heff(\tau)\rvert^2
    \de \tau\right\rvert<\epsilon_1
\end{align}
for all $T<T_1$.

\parmark{2. Approximation} For every $\delta\in\mysetR$, let $\heff^{(\delta)}$
denote the translation of $\heff$ defined by
$\heff^{(\delta)}(\tau)=\heff(\tau-\delta)$, $\tau\in\mysetR$. According
to \cite[Theorem 9.5]{Rudin1987}, the mapping
\begin{align}
\delta\mapsto \heff^{(\delta)}
\end{align}
is uniformly continuous in $\ltwoR$. Together with the reverse triangular
inequality \eqref{eq:reverseTriangular}, this implies that for every
$\epsilon_2>0$,
there exists a $\delta_1$ such that
  \begin{align}\label{synchronization:secondApproximation}
    \left\lvert\int\limits_0^{D_s}\lvert\heff(\tau)\rvert^2
      \de \tau -\int\limits_0^{D_s} \lvert\heff(\tau-\delta)\rvert^2
      \de \tau \right\rvert<\epsilon_2
  \end{align}
for all $\delta<\delta_1$.

\parmark{3. Approximation} Since $\heff$ is continuous (by convention
\eqref{eq:continuous})
and since we consider it over a compact interval, the term
\begin{align*}
 \sum\limits_{l=0}^{\lfloor\frac{D_s}{T'}\rfloor-1}T'\lvert
 \heff(lT'-\delta)\rvert^2
\end{align*}
is equal to the Riemann sum of $\lvert h_T(\tau-\delta)\rvert^2$ over the
interval $[0,D_s]$. For $T'\rightarrow 0$, the Riemann sum converges to the
integral. Therefore, for every $\epsilon_3>0$, there exists a $T'_1$ such that
  \begin{align}
\left\lvert\int\limits_0^{D_s} \lvert\heff(\tau-\delta)\rvert^2\de
\tau-\hspace{-0.2cm}\sum\limits_{l=0}^{\lfloor\frac{D_s}{T'}\rfloor-1}\hspace{
-0.1cm} T'\lvert
 \heff(lT'-\delta)\rvert^2\right\rvert\hspace{-0.1cm}<\epsilon_3
\end{align}
for all $T'<T'_1$. We define $\epsilon_0=\epsilon_1+\epsilon_2+\epsilon_3$ and
$T_0=\min\{T_1,T'_1,\delta_1\}$. Combining the three approximations, we
have shown that for every $\epsilon_0>0$, there exists a $T_0$ such that
\begin{align}
\left\lvert\lVert
h\rVert^2-\sum\limits_{l=0}^{\lfloor\frac{D_s}{T'}\rfloor-1} T'\lvert
 \heff(lT'-\delta)\rvert^2\right\rvert<\epsilon_0
\end{align}
for all $T,T',\delta < T_0$. This result is also valid for $T'=T$ and $\delta$
and $T$ jointly tending to zero with $0\leq \delta <T$. This concludes the
proof.
\end{proof}
For a further discussion of this result we refer to
Section~\ref{sec:simulation}.

%% file: simulation.tex
\section{Simulation}\label{sec:simulation}
\begin{center}
\begin{table*}
\caption{Simulation Results for $\dspread=279\nanos$}
\vspace{0.2cm}
\begin{center}
 \begin{tabular}{lll|rrrrrrrrr}
\hline
\parbox[b]{0cm}{\vspace{0.06cm\ }}&\hspace{-0.4cm}$W$&[MHz]
&\hspace{-0.4cm}4&8&16&32&64&128&256&512&1024\\
&\hspace{-0.4cm}$\powerPenalty$&[\%]&68.9&48.0&25.8&11.7&5.8&2.5&1.9&0.9&0.4\\
&\hspace{-0.4cm}$\bar{\powerPenalty}$&[\%]
&46.8&16.8&10.0&4.3&2.1&0.9&0.4&0.2&0.1\\
\multicolumn{3}{l|}{\hspace{0.05cm}$L=\lfloor\dspread
W\rfloor$}&1&2&4&8&17&35&71&142&285
\\[0.1cm]
\hline 
\end{tabular}
\end{center}
\label{tab:simulationResults}
\end{table*}
\end{center}
We simulate the effect of the timing offset $\delta$ onto the overall gain
of channel candidates by calculating $\lVert\htrunc_T^{(\delta)}\rVert^2$
for increasing bandwidths $W$. As a data set, we use $100$ impulse responses
according to the IEEE 802.15.4a UWB channel model defined in \cite{Molisch2004}
by using the MATLAB script \url{uwb_sv_eval_ct_15_4a.m}, which is also provided
in \cite{Molisch2004}. The generated impulse responses have a normalized delay
spread of $\dspread=279\mathrm{ns}$. Since the data set already comes in digital
form and since we are limited to digital signal processing in simulation, it is
difficult to compare all possible channel candidates for $\delta\in[0,T)$,
$\delta$ continuous. We therefore resort to the following. We consider a generic
all digital receiver that has a small scale timing synchronizer with a time
resolution of $T/4$. With this receiver in mind, we compare the channel
candidates for every impulse response $h^{(i)}$ from the data set and every
considered bandwidth $W=1/T$ in the following way.

\parmark{Step 1} We lowpass filter $h^{(i)}$ to obtain $\heff^{(i)}$.

\parmark{Step 2} We randomly generate a small timing offset $\varepsilon$
uniformly distributed over $[0,T/4)$. The generic receiver in mind will not
resolve $\varepsilon$.

\parmark{Step 3} For $m=0,\dotsc,3$, we consider the sequence
\begin{align}
\heff^{(i)}\lefto(lT-\varepsilon-\frac{mT}{4}\right)
\end{align}
and calculate the channel gain of the truncated impulse response to
\begin{align}
\lVert \htrunc_T^{(i,m)}\rVert^2=\max_{k}\sum\limits_{l=0}^{\lfloor
\frac{D_s}{T}\rfloor-1}T\lvert\heff^{(i)}(kT+lT-\varepsilon-\frac{mT}{4}
)\rvert^2.
\end{align}
At the receiver, the maximization over $k$ corresponds to an energy based large
scale timing acquisition. The maximum relative channel gain loss when
synchronizing without a small scale timing synchronizer is given by
\begin{align}
 \powerPenalty_{T,\max}^{(i)}=1-\frac{\min_{m}\lVert
\htrunc^{(i,m)}\rVert^2}{\max_{m}\lVert \htrunc^{(i,m)}\rVert^2}.
\end{align}
\begin{figure}
\small
\centering
\psfrag{Penalty [
\psfrag{0}{$0$}
\psfrag{10}{$10$}
\psfrag{20}{$20$}
\psfrag{30}{$30$}
\psfrag{40}{$40$}
\psfrag{50}{$50$}
\psfrag{60}{$60$}
\psfrag{70}{$70$}
\psfrag{maximum_penalty}{maximum penalty $\powerPenalty_T$}
\psfrag{average_penalty}{average penalty $\bar{\powerPenalty}_T$}
\psfrag{10MHz}{$10$}
\psfrag{100MHz}{$100$}
\psfrag{1000MHz}{$1000$}
\psfrag{Bandwidth}{Bandwidth $[\mhz]$}
\includegraphics[width=\linewidth]{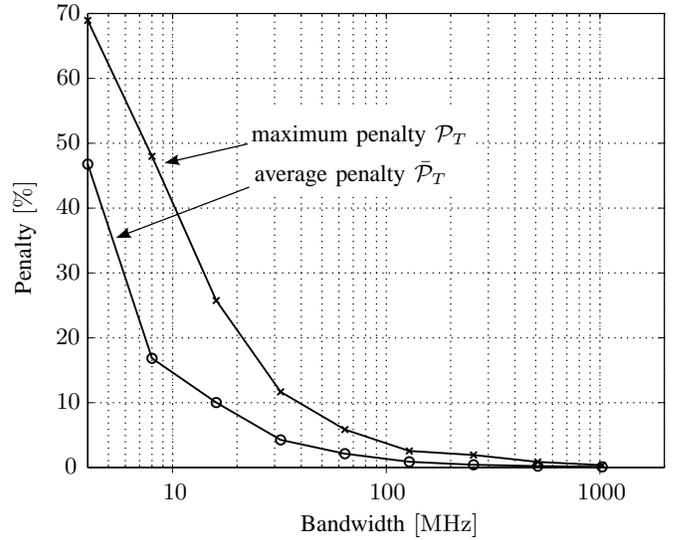}
\caption{The average relative penalty $\bar{\powerPenalty}_T$ and the maximum
relative penalty $\powerPenalty_T$ for timing acquisition without
small scale timing synchronization as a function of the bandwidth $W=1/T$.}
\label{fig:powerRatios}
\end{figure}
We assess the simulation results by considering both the maximum relative
performance penalty and the average relative performance penalty given by
\begin{align}
\powerPenalty_T= \max_i
\powerPenalty_{T,\max}^{(i)}\quad\text{and}\quad\bar{\powerPenalty}_T=
\frac{1}{100}\sum\limits_{i=1}^{100}\powerPenalty_{T,\max}^{(i)}
\end{align}
as functions of the sampling time $T=1/W$. In Table~\ref{tab:simulationResults},
we provide the simulation results consisting of the corresponding values
for the bandwidth, the worst case penalty, the average penalty, and the number
of channel taps. As can be seen from
Figure~\ref{fig:powerRatios}, both the worst case penalty and the average
penalty decrease monotonically with an increasing bandwidth. Both
curves converge to $0$, which corresponds to our analytic result from the
previous section.

We observe that, roughly for $L=\lfloor \dspread W\rfloor>10$,
synchronizing with a resolution of $T$ is sufficient in terms of channel gain,
since the maximum penalty is smaller than $5\%$. On the other hand, in a
narrowband scenario with $L=\lfloor \dspread W\rfloor<4$, the penalty due to a
missing small scale synchronizer can be over $50\%$.

%% file: conclusions.tex
\section{Conclusions}
We have shown that for communications strictly bandlimited to $W$ over 
an LTI channel with delay spread $\dspread$, the truncated discrete time channel
model
\begin{align}
  y[k]=\sum\limits_{l=0}^{L-1} h[l]x[k-l],\quad L\approx
\lfloor \dspread W\rfloor
\end{align}
is robust against unknown timing offsets $\delta\in[0,1/W)$ between
sender and receiver in the limit when the bandwidth $W$ of the considered system
goes to infinity. A simulation shows that this theoretic result is valid for the
IEEE 802.15.4a channel model for $W>50$MHz. As a receiver design criterion, this
means that small scale synchronization with a precision higher than the sampling
time $T=1/W$ is not necessary for UWB communication systems. As an extension of
our work, it may be of interest to investigate if the robustness against the
small scale timing offset $\delta$ remains when drift compensation, large scale
timing acquisition and channel estimation in the design of a wideband receiver
are considered jointly.

%% file: acknowledgement.tex
\section*{Acknowledgment}
This work was partly supported by the Deutsche Forschungsgemeinschaft (DFG)
project UKoLoS (grant MA 1184/14-1) and the UMIC excellence cluster of RWTH
Aachen University.